\documentclass[12pt]{article}
\newtheorem{proposition}{Proposition}

\newtheorem{lemma}{Lemma}

\newenvironment{proof}{\noindent{\bf Proof:}}{\hfill\fbox{}\vspace*{1mm}}
\usepackage[usenames]{color}
\usepackage{graphicx}
\begin{document}
\title{\bf On Reduced Form Intensity-based Model with Trigger Events}
\author{Jia-Wen Gu
\thanks{Advanced Modeling and Applied Computing Laboratory,
Department of Mathematics, The University of Hong Kong,
Pokfulam Road, Hong Kong. Email:jwgu.hku@gmail.com.
}
\and Wai-Ki Ching
\thanks{Corresponding author. Advanced Modeling and Applied Computing Laboratory,
Department of Mathematics, The University of Hong Kong, Pokfulam
Road, Hong Kong. E-mail: wching@hku.hk. Research supported in
part by RGC Grants 7017/07P, HKU CRCG Grants
and HKU Strategic Research Theme Fund on Computational Physics and Numerical Methods.}
\and Tak-Kuen Siu
\thanks{ Department of Applied Finance and Actuarial Studies,
Faculty of Business and Economics, Macquarie University,
Macquarie University, Sydney, NSW 2109, Australia. Email: ken.siu@mq.edu.au,
ktksiu2005@gmail.com,}
\and Harry Zheng
\thanks{ Department of Mathematics,
Imperial College, London, SW7 2AZ, UK. Email: h.zheng@imperial.ac.uk.}
}
\date{}
\maketitle

Corporate defaults may be triggered by some major market news or events such as financial crises or collapses of
major banks or financial institutions.
With a view to develop a more realistic model for credit risk analysis,
we introduce a new type of reduced-form intensity-based model
that can incorporate the impacts of both observable ``trigger'' events and economic environment on corporate defaults.
The key idea of the model is to augment a Cox process with trigger events. Both  single-default and multiple-default cases are considered in this paper.
In the former case, a simple expression for the distribution of the default time is obtained.
Applications of the proposed model
to price defaultable bonds and multi-name Credit Default Swaps (CDSs) are provided. \\


\newpage
\section{Introduction}

Modeling default risk has long been an important problem in
both theory and practice of banking and finance.
In the aftermath of the global financial crisis (GFC), much attention has been paid to investigating the appropriateness of the current practice of default risk modeling
in banking, finance and insurance industries.
Popular credit risk models currently used in the industries have their origins in two major classes of
models. The first class of models was pioneered by Black and Scholes (1973) and Merton (1974) and is called a structural firm value model.
The basic idea of the model is to describe explicitly the relationship between the asset value of a firm and the default of the firm.
More specifically, the default of the firm is triggered by the event that the asset value of the firm falls below a certain threshold
level related to the liabilities of the firm.
The structural firm value model provides the theoretical basis for the commercial KMV model
which has been widely used for default risk model in the financial industry.
The second class of models was developed by Jarrow and Turnbull (1995) and Madan and Unal
(1998) and is called a reduced-form, intensity-based credit risk model.
The basic idea of the model is to consider defaults as exogenous events and to model
their occurrences using Poisson processes and their variants.
In this paper, we focus on reduced-form, intensity-based credit risk models.

Reduced-form, intensity-based credit risk models have been widely used to model
portfolio credit risk and to describe dependent default risks. There are two major types of reduced-form, intensity-based models
for describing dependent default risk, namely bottom-up models and top-down models.
Bottom-up models focus on modeling default intensities of
individual reference entities and their aggregation to form
a portfolio default intensity.  Some works on bottom-up models
include Duffie and Garleanu (2001), Jarrow and Yu (2001),
Sch\"onbucher and Schubert (2001), Giesecke and Goldberg (2004),
Duffie, Saita and Wang (2006) and Yu (2007) etc.
 These works differ mainly in their specifications
for the parametric forms of default intensities of individual
entities and the way these intensities are aggregated.
The top-down models concern modeling the occurence
defaults at a portfolio level. A default intensity for the whole portfolio
is modeled without reference to the identities of individual entities.
Some procedures such as random thinning can be used to recover the default
intensities of the individual entities.
Some works on top-down models include Davis and Lo (2001),
Giesecke, Goldberg and Ding(2011),
Brigo, Pallavicini and Torresetti (2006),
Longstaff and Rajan (2008) and Cont and Minca (2011).

We focus on the bottom-up model. Lando (1988) proposed a reduced-form, intensity-based model, where the occurrence of a default is described by the first jump of a Cox process.
The main advantage of the Lando's model is that under his model,
a simple pricing formula for a defaultable risky asset can be obtained.
This formula is similar to the one for the default-free counterpart of the risky asset.
Yu (2007) extended the Lando's model to incorporate multiple defaults and their correlation. The so-called ``total hazard construction''
by Norros (1986) and Shaked and Shathanthikumar (1987)
was used to generate default times with interacting intensities.
Zheng \& Jiang (2009) proposed a unified factor-contagion model for modeling
correlated defaults and provide an analytical solution for modeling default
times with ``total hazard construction''.
Gu et al. (2011) introduced an ``ordered default rate'' method to give a recursive formula
for the distribution of default times in pricing basket Credit Default Swaps (CDSs) in the context of a reduced-form, intensity-based model, which significantly enhances the computational efficiency in finding the prices of CDSs.
One of the shortcomings in a number of existing reduced-form intensity-based models is that they fail to incorporate the impact of major market events, such as financial crises, on corporate defaults.
This may lead to underestimation of default risk and also undervaluation
of defaultable risky products.

In this paper, we address the problem of how to incorporate the direct impact of observable ``trigger'' events
on corporate defaults in a reduced-form, intensity-based credit risk model.
The key idea is to describe the occurrence
of a default by the first unrecoverable ``trigger'' event. Armed with a Cox process for default risk, we incorporate
the impact of economic environment on defaults by allowing the default intensity depending on an underlying
state process representing the variation of economic environment over time. We consider both the single-default
and multiple-default cases. In the single-default case, we obtain a simple expression for the distribution of the
default time. This distribution is useful for pricing defaultable securities. We then extend the model to the multiple-default
case with a view to incorporating default correlation. To provide a tractable and practical way to value defaultable
securities, we focus on the case where the state process for economic environment is modeled by a continuous-time,
finite-state, observable Markov chain. Applications of the proposed model to value defaultable bonds and
basket Credit Default Swaps (CDSs) are discussed. We also provide numerical results to illustrate the sensitivity of the
prices of these securities with respect to changes in key parameters.

The rest of the paper is organized as follows.
Section 2 presents the basic model and derives the distribution of a default time.
Section 3 provides the extension of model framework
to the case of multiple correlated defaults.
Section 4 presents the Markov chain model for the state process of
economic environment.  Applications to pricing the defaultable
securities are given in Section 5.
We then  conclude the paper in Section 6.

\section{The Basic Model}

A popular reduced-form, intensity-based credit risk model was proposed by Lando
(1998), where the occurrence of a default was described by the first jump of a Cox process
with stochastic intensities $\{ \lambda_t \}_{t \geq 0}$
depending on an underlying state process $\{ X_t  \}_{t \geq 0}$
describing the evolution of economic environment over time.
Here we aim at extending the Lando's model by incorporating explicitly the
``trigger'' events such as financial crises or extra-ordinary market news into default
risk modeling. We assume that these ``trigger'' events are observable and may lead
to the default of a corporation. Furthermore, we suppose that the corporation may recover from
a ``trigger'' event via re-organizing its resources or re-structuring.
We assume that the emergence of ``trigger'' events are modeled by a Cox process,
which is also called doubly Poisson process in the statistical literature and has a remarkable
history in statistics. In what follows, we shall describe the mathematical set up of the
Cox process describing the ``trigger'' events.

Uncertainty is described by a complete filtered probability space
$(\Omega, {\cal F}, \{ {\cal F}_t \}_{t \geq 0}, P)$, where $P$ is a given probability measure
\footnote{ When we wish to evaluate the risk of a credit portfolio, we need to use a
real-world probability measure.
In this case, $P$ can be interpreted as a real-world probability measure.
On the other hand,  when we wish to price defaultable securities, we must use a risk-neutral probability measure.
In this case, there are two approaches to interpret $P$.
The first approach is to interpret $P$ as a risk-neutral probability
measure and start with the risk-neutral probability measure directly. The second approach
is to interpret $P$ as a real-world probability measure and then use a measure change
for Poisson processes to transform the real-world probability measure to a risk-neutral
one.
To simplify our discussion, when we discuss the pricing of defaultable securities,
we shall adopt the first approach.} and $\{ {\cal F}_t \}_{t \geq 0}$  is a filtration
satisfying some usual conditions, (i.e., the right-continuity and the $P$-completeness).
We shall define precisely the filtration $\{ {\cal F}_t \}_{t \geq 0}$ in later part of
this section.

To describe the evolution of the state of economic environment
over time, we define a state process $\{ X_t \}_{t \geq 0}$. 
We assume that
$\{ X_t \}_{t \geq 0}$ is a c\`adl\`ag, (i.e., right continuous with left limits), process on
$(\Omega, {\cal F}, \{ {\cal F}_t \}_{t \geq 0}, P)$ with state space $\Re$.
Let $\{ N_t \}_{ t \geq 0 }$ be a standard Poisson process on
$(\Omega, {\cal F}, \{ {\cal F}_t \}_{t \geq 0}, P)$, with $N_0 = 0$, $P$-a.s.
Write $\{ \lambda_t \}_{t \geq 0}$ for a bounded, non-negative stochastic
process on  $(\Omega, {\cal F}, \{ {\cal F}_t \}_{t \geq 0}, P)$.
We assume that for each $t \geq 0$, $\lambda_t :=\lambda (X_t)$,
for some non-negative continuous function $\lambda$
and that $\{ N_t \}_{t \geq 0}$ and $\{ \lambda_t \}_{t \geq 0}$
are stochastically independent under $P$. 
For each $t \geq 0$, we define the following cumulative process $\{ \Lambda_t \}_{t \geq 0}$:
\begin{eqnarray*}
\Lambda_t := \int^{t}_{0} \lambda_s d s < \infty \ .
\end{eqnarray*}
Then a Cox (point) process $\{ \tilde{N}_t \}_{t \geq 0}$ with intensity measure $\Lambda := \{ \Lambda_t \}_{t \geq 0}$
\footnote{ Strictly speaking, the intensity measure $\Lambda$ is defined on the $\sigma$-field generated by
bounded subsets of the interval $[0, \infty)$. The intensity measure $\Lambda$ also has a Random density, or Radon-Nikodym
derivative, $\lambda$, by definition.} is defined by:
\begin{eqnarray*}
{\tilde N}_t := N _{\Lambda_t} \ .
\end{eqnarray*}
For each $i = 1, 2, \cdots$, let $\tau^i$ be the arrival time of the $i^{th}$ ``trigger'' event, which is modeled as the arrival time of a jump in the Cox process.
Once a ``trigger'' event occurs at time $s$, a loss occurs to the firm , which is modeled as an arbitrary independent random variable ``$L$ ''.
We write $\{ C_t \}_{t \geq 0}$ for the process depending
on the state process $\{ X_t  \}_{t \geq 0}$,
where $C_t := C (X_t)$ as the threshold value at time $t$, for some non-negative continuous function $C$.  If $L \leq C_s$, then the firm can recover from the ``trigger'' event; otherwise, the firm defaults.
 Davis \& Lo (2001) introduce a Bernoulli contagion variable in a homogeneous setup, which is similar to  variable ``$L$'' Here.

Let $\{ \tau^i  \}_{i = 1, 2, \cdots}$ be
a sequence of stopping times representing the arrival times of ``trigger'' events defined by
\begin{eqnarray*}
\tau^i := \inf \{ t \geq 0 : \tilde{N}_t \geq i \} \ ,
\end{eqnarray*}
and $\{ L^i \}_{i = 1, 2, \cdots}$ a sequence of arbitrary independent and identically distributed random variables. 
We assume that $\{ \tau^i  \}_{i = 1, 2, \cdots}$ and $\{ L^i \}_{i = 1, 2, \cdots}$ are stochastically independent under $P$.
Define a random variable $K$ taking values in $\{ 1, 2, \cdots \}$ by
\begin{eqnarray*}
K := \min \{ i : L^i > C_{\tau^i} \} \ .
\end{eqnarray*}
Then the default time of a firm $\tau$ is defined by: $\tau := \tau^K$.
We now specify more explicitly the information structure of our model.
Define the filtrations $\{ \mathcal{G}_t \}_{t \geq 0}$, $\{ \mathcal{H}_t \}_{t \geq 0}$ and
$\{ \mathcal{I}_t \}_{t \geq 0}$ as follow:
$$
\mathcal{G}_t := \sigma \{ X_s : 0 \leq s\leq t \} \vee {\cal N},
$$
$$
\mathcal{H}_t:= \sigma \{ \tilde{N}_s: 0\leq s \leq t \} \vee {\cal N} \ ,
$$
and
$$
\mathcal{I}_t := \sigma \{ 1_{\{\tau \leq s\}} : 0 \leq s\leq t \} \vee {\cal N},
$$
where ${\cal N}$ is the collection of all $P$-null subsets in ${\cal F}$ and
$1_A$ is the indicator function of an event $A$.
Here we assume that the filtration $\{ {\mathcal{F}}_t \}_{t \geq 0}$
is specified as follows:
$\mathcal{F}_t := \mathcal{G}_t \vee \mathcal{H}_t \vee \mathcal{I}_t$.
This represents the full observable information structure in our model.

For each $s \geq 0$, let $q_s$ be the probability that the firm can recover from a ``trigger'' event if the event occurs at time $s$ given the underlying state $X_s$.
Then,
\begin{eqnarray*}
q_s := P (L \leq C(X_s) \mid X_s) .
\end{eqnarray*}
Let $p_s := 1 - q_s$. Then
\begin{eqnarray*}
p_s = P (L > C(X_s) \mid X_s)  .
\end{eqnarray*}
The following result is one of our main results which gives the conditional and unconditional distributions of the default time $\tau$.
The proof can be found the Appendix.

\begin{proposition} For any $t \geq s > 0$,
$$
P(\tau >s \mid \mathcal{G}_t) = \exp\left\{-\int_0^s p_u \lambda_u du\right\}
$$
and
$$
P(\tau >s) = E\left[\exp\left\{-\int_0^s p_u \lambda_u du\right\}\right] \ ,
$$
where $E$ is an expectation under $P$.
\end{proposition}

The following results are also important to
characterize the probability laws of the default time $\tau$.
Their proofs can be found in the Appendix.

\begin{lemma}
For any $s<t$,
$$
E(1_{\{ \tau > t \}} \mid \mathcal{G}_t \vee \mathcal{H}_s \vee \mathcal{I}_s ) = 1_{\{ \tau >s\}}\exp\left\{-\int_s^t p_u \lambda_u du\right\}.
$$
\end{lemma}

\begin{lemma} The process
$$
1_{\{t \geq \tau\}} - \int_0^t p_u \lambda_u 1_{\{u < \tau\}} du \ , \quad t \geq 0 \ ,
$$ is an $( \{ {\cal F}_t \}_{t \geq 0}, {\cal P})$-martingale.
\end{lemma}

To price the defaultable securities in the proposed model as above,
we construct three ``building blocks'' as in Lando(1988).
Before that, we assume that all of the expectations
in this paper are taken under an equivalent martingale measure.
We suppose that $T$ denotes the expiry date of all contingent claims.
The three  ``building blocks'' are as follows:
\begin{itemize}
\item[(I)] $X1_{\{ \tau > T\}}$: A payment $X \in \mathcal{G}_T$ at a fixed date $T$ which occurs if there has been no default before time $T$.
\item[(II)] $Y_s 1_{\{\tau >s\}}$: A stream of payments at a rate specified by the $\{ \mathcal{G}_t \}_{t \geq 0}$-adapted process $Y$ which stops when default occurs.
\item[(III)] $Z_{\tau}$: A recovery payment at the time of default, where $Z$ is a $\{\mathcal{G}_t\}_{t \geq 0}$-adapted process.
\end{itemize}
Now we proceed to give the pricing formula of these three ``building blocks''.

\begin{proposition}
Suppose
$$
\begin{array}{ll}
\exp\{-\int_0^T r_s ds\}X,\\
\int_0^T Y_s \exp\{-\int_0^s r_u du\} ds,  \\
\int_0^T Z_s\lambda_s p_s \exp\{-\int_0^s (r_u+\lambda_up_u) du\} ds
\end{array}
$$
are integrable random variables. Then,
\begin{equation}
\begin{array}{lll}
E(\exp\{-\int_0^T r_s ds\}X 1_{\{\tau > T\}})&=&E(\exp\{-\int_0^T(r_s +p_s\lambda_s)ds\}X),
\end{array}
\end{equation}

\begin{equation}
\begin{array}{lll}
 E(\int_0^T Y_s 1_{\{\tau >s\}}\exp\{-\int_0^s r_u du\} ds)
&=&E(\int_0^T Y_s \exp\{-\int_0^s (r_u+p_u\lambda_u) du\} ds)
\end{array}
\end{equation}
and
\begin{equation}
\begin{array}{lll}
E(\exp\{-\int_0^{\tau} r_s ds\}Z_{\tau})
&=&E(\int_0^T Z_s\lambda_s p_s \exp\{-\int_0^s (r_u+\lambda_up_u) du\} ds).
\end{array}
\end{equation}

\end{proposition}
\begin{proof}
The results follow directly from Proposition 1.
\end{proof}

\section{An Extension to Dependent Multiple defaults}

In this section, we shall extend the basic model in the last section
to the multiple-default case.
Again we consider the filtered probability space
$(\Omega, {\cal F}, \{ {\cal F}_t \}_{t \geq 0}, P)$,
where the filtration satisfies the usual conditions and is specified as in the last section. Here the underlying state process of economic environment $\{ X_t \}_{t \geq 0}$ is also defined as in the last section, (i.e., a c\`adl\`ag process).
For each $i = 1, 2, \cdots, n$, let $\tau_i$ be a stopping time representing the default time of the $i^{th}$ individual obliger.

Let $\{ {N}^i_t \}_{t \geq 0}$, $i =1, 2, \cdots, n$, be $n$ independent standard Poisson processes with $N^i_0=0$,
$P$-.a.s. Write, for each $i = 1, 2, \cdots, n$, $\{ \lambda^i_t \}_{t \geq 0}$ for a bounded, non-negative stochastic
process which is adapted to the enlarged filtration containing the
filtration $\{ {\cal G}_t \}_{t \geq 0}$ and
the filtration generated by the sequence of stopping times $\{ \tau_i \}_{i = 1, 2, \cdots, n}$, where the
latter filtration will be defined precisely later in this section.

For each $i = 1, 2, \cdots, n$ and each $t \geq 0$, we define the intensity measure $\Lambda^i := \{ \Lambda^i_t \}_{t \geq 0}$
for the $i^{th}$-individual obliger by putting:
$$
\Lambda^i_t := \int_0^t \lambda^i_s ds < \infty \ . 
$$

Then for each $i = 1, 2,, \cdots, n$, a Cox (point) process $\{ {\tilde N}^i_t \}_{t \geq 0}$ for the
$i^{th}$ obligor associated with the intensity measure $\Lambda_i$  is defined by:
\begin{eqnarray*}
\tilde{N}^i_t := N^i_{\Lambda^i_t } \  .
\end{eqnarray*}
Define the arrival time of $j^{th}$ ``trigger'' event of name $i$,
$$
\tau^{j}_i=\inf\{t: \tilde{N}^i_t \geq j\}, i=1, 2, \ldots, n.
$$
Consider an array of i.i.d. arbitrary random variables
$$
\{L^j_i, i=1,2, \ldots,n, j=1,2,\ldots\}.
$$
We suppose that this array of 
random variables is independent of all previous events.
Define the random variables:
$$
K_i := \min\{ j : L^{j}_i > C^j_{\tau^j_i}\} \ , \quad i = 1, 2, \cdots, n,
$$
where $C^i_s := C^i(X_s)$ is the threshold value of name $i$ at time $s$ and $C^i$ a non-negative continuous function for $i=1,2, \ldots,n$.

Then the default time of  name $i$ is defined as
$$
\tau_i :=\tau^{K_i}_i.
$$
The information structure of the multiple-default model is specified as follows:
$$
\mathcal{H}_t:= \sigma \{ (\tilde{N}^i_s)_{i=1}^n: 0\leq s \leq t \} \vee {\cal N} \ ,
$$
and
$$
\mathcal{I}_t
:= \sigma \{ (1_{\{\tau_i \leq s\}})_{i=1}^n : 0 \leq s\leq t \} \vee {\cal N} \ .
$$
Again we assume that
$$
\mathcal{F}_t=\mathcal{G}_t \vee \mathcal{H}_t \vee \mathcal{I}_t.
$$
We also define 
$$
p^i_s :=P(L > C^i(X_s)\mid X_s), \ \ q^i_s := P(L \leq C^i(X_s)\mid X_s).
$$
We further assume that for each $i = 1, 2, \cdots, n$, $P(0<\tau_i<\infty)=1$ and $P(\tau_i=\tau_j)=0$ for any $i \neq j$.
Moreover, $X$ is an ``exogenous'' stochastic process in the sense that $\mathcal{G}_\infty$ and $\mathcal{H}_t\vee \mathcal{I}_t$ are conditionally independent
given $\mathcal{G}_t$.

We remark that this framework for multi-name defaults allows the default intensities to be sensitive to the observed defaults as well as the underlying state process $X$.
Taking the consideration of the real practice that once a firm recovers from a ``trigger'' event, the impact of such a ``trigger'' is restricted to a very minor level,
$\{ \lambda_t \}_{t \geq 0}$ is supposed $\{ \mathcal{G}_t \vee \mathcal{I}_t\}_{t \geq 0}$-adapted.

For each $i = 1, 2, \cdots, n$, if we treat information about the state process $X$ and observed defaults of other names up to time $t$ as a new ``$\mathcal{G}_t$'',
then applying Lemma 2 in the single-name default case,
we deduce that the process defined by:
$$
1_{\{t \geq \tau_i\}}-\int_0^t p^i_u \lambda^i_u 1_{\{u < \tau_i\}}du \ , \quad t \geq 0 \ ,
$$
is an $(\{ {\cal F}_t \}_{t \geq 0}, P)$-martingale.

Consequently, the total hazard construction method, Yu(2007) and Zhang \& Jiang (2009),
the order default rate approach by Gu et al.(2011) can be applied to compute the
multi-name default time distribution under this framework.
We also remark that this framework is indeed a generalization of the standard reduced-form intensity-based model by Lando (1988) and Yu (2007).
If we set $p^ i_s \equiv 1$, i.e., the firm cannot recover from ``trigger'' event $P$-a.s.,
then our proposed model is reduced to be the standard reduced-form intensity-based model.

\section{State Process as an Observable Markov Chain}

Suppose now that the state process $(X_t)_{t \geq 0}$ follows a continuous-time, homogeneous, $M$-state Markov chain
on $(\Omega, {\cal F}, P)$ with state space $\{x_1, x_2, \ldots,x_M\}$. The amount of time the chain $X$
spends in state $x_i$ before making a transition into another state is exponentially distributed with rate $v_i$ .
When the process leaves state $x_i$, it next enters state $x_j$ with probability $p_{ij}$,
where
$$
p_{ii}=0 \quad  {\rm and} \quad \sum_j p_{ij}=1
$$
for all $i$.

Let $T_{ij}(t)$ be the occupation time of the chain $X$ in state $x_j$ in the time interval
$[0, t]$ starting from $X_0 =x_i$.  We wish to determine the joint distribution of $(T_{i1}(t), T_{i2}(t), \ldots, T_{iM}(t))$.
Note that the joint distribution of $(T_{i1}(t), T_{i2}(t), \ldots, T_{iM}(t))$ is completely determined by
its joint moment generating function. We shall derive the joint moment generating function in the
sequel.

For each $i = 1, 2, \cdots, M$, let
$$
T_i(t)=(T_{i1}(t), T_{i2}(t), \ldots, T_{iM}(t))^T
$$
and
$$
u=(u_1,u_2, \ldots, u_M)^T \in \mathcal{R}^M.
$$
The moment generating function of $T_i(t)$ is given by:
$$
\Psi_i(u,t)=E(\exp\{u^{T} T_i(t)\})
$$
Let $\xi_i$ denote the time of the first jump from state $x_i$ to another state, i.e., $\xi_i \sim \exp(v_i)$.
Hence
$$
\begin{array}{lll}
\Psi_i(u,t)&=& \displaystyle E(\exp\{u^{T} T_i(t)\})\\
&=&\displaystyle E(E(\exp\{u^{T} T_i(t)\} \mid \xi_i))\\
&=&\displaystyle \sum_{k \neq i} p_{ik} v_i\int_0^t e^{(u_i-v_i)s} E(\exp\{u^T T_k(t-s)\}) ds + e^{(u_i-v_i)t}\\
&=&\displaystyle \sum_{k \neq i} p_{ik} v_i\int_0^t e^{(u_i-v_i)s} \Psi_k(u, t-s) ds + e^{(u_i-v_i)t}\\
&=&\displaystyle \sum_{k \neq i} p_{ik} v_i\int_0^t e^{(u_i-v_i)(t-s)} \Psi_k(u, s) ds + e^{(u_i-v_i)t}.
\end{array}
$$
Taking the partial derivative with respect to $t$ on both sides yields,
\begin{equation}\label{4p}
\frac{\partial}{\partial t}\Psi_i(u,t)= \sum_{k \neq i} p_{ik} v_i \Psi_k(u, t) +(u_i-v_i) \Psi_i(u,t).
\end{equation}
For simplicity,
we write
$$
\Psi_u(t)=(\Psi_1(u,t), \Psi_2(u,t), \ldots, \Psi_M(u,t))^T
$$
and define the matrix
$$
A=\left[
\begin{array}{cccc}
u_1-v_1   & p_{12}v_1 & \cdots &p_{1M} v_1\\
p_{21}v_2 & u_2-v_2   & \cdots &p_{2M} v_2\\
\vdots    &  \vdots   & \ddots &\vdots\\
p_{M1}v_M & p_{M2}v_M & \cdots &u_M-v_M\\
\end{array}
\right]
$$
Hence Equation (\ref{4p}) can be rewritten as a system of homogeneous linear differential equations with constant coefficients
\begin{equation}\label{4pp}
{\Psi_u}^{\prime}(t) = A \Psi_u(t),
\end{equation}
where $\Psi_u(0)={\bf 1}$. Solving the ODEs,  by using the Fundamental Theorem for Linear Systems (see Chapter 1 in L. Perko (2001)),  yields the following
proposition.

\begin{proposition}
The moment generating function of $T_i(t)$ is given by $\Psi_i(u,t)$, where $$
\Psi_u(t)=(\Psi_1(u,t), \Psi_2(u,t), \ldots, \Psi_M(u,t))^T
$$
has a unique solution as
$$
\Psi_u(t)= e^{At}{\bf 1}.
$$
\end{proposition}

\noindent
This result will be used to price basket credit defaut swaps under our proposed model in
Section 5.2.

\section{Applications}
In this section, we apply the ordered default rate approach
in  Gu et al. (2011) to price different defaultable securities
under various assumptions of default correlation structures.

\subsection{Defaultable Bonds}

We first discuss the pricing of defaultable zero-coupon bonds with zero recovery under a constant default-free rate $r$
in the case of two firms.  The defaultable bond price is then proportional to the conditional survival probability $P(\tau_i>T \mid \mathcal{F}_t)$.
A ``looping default'' case with two firms was presented by Jarrow and Yu (2001), where they assumed that the default times of the two
firms  $\tau^A$ and $\tau^B$ have the following intensities:
$$
\lambda^A_t =a_1+a_2 1_{\{t \geq \tau^B\}}
$$
and
$$
\lambda^B_t =b_1+b_2 1_{\{t \geq \tau^A\}}.
$$
Yu (2007) gave the default time distribution of this two-firm case in the standard reduced-form intensity-based modelling framework.
In this section, we provide a solution of the two-firm case in our proposed modeling framework, where $p^i_t \equiv p$.
In other words, the firm recover from a ``trigger'' event with a constant probability $1-p$.

Let $\tau^A \wedge \tau^B$ denote the first default time of these two names, the the default rate of $\tau^A \wedge \tau^B$ is $a_1+b_1$.
By Proposition 1,
$$
P(\tau^A \wedge \tau^B >t)= \exp(-p(a_1+b_1)t).
$$
Hence
$$
\left\{
\begin{array}{ll}
P(\tau^A >t, \tau^A < \tau^B) =\displaystyle \frac{a_1}{a_1+b_1} e^{-p(a_1+b_1)t}\\
P(\tau^B >t, \tau^A > \tau^B) =\displaystyle \frac{b_1}{a_1+b_1} e^{-p(a_1+b_1)t}
\end{array}
\right.
$$
and also
$$
\left\{
\begin{array}{ll}
P(\tau^A-\tau^B>t \mid \tau^A> \tau^B) =e^{-p(a_1+a_2)t}\\
P(\tau^B-\tau^A>t \mid \tau^B> \tau^A) =e^{-p(b_1+b_2)t}.
\end{array}
\right.
$$
Then by making use of convolution, we have
$$
\left\{
\begin{array}{ll}
\displaystyle \frac{\partial}{\partial t}P(\tau^A \leq t, \tau^A>\tau^B)= \frac{b_1(a_1+a_2)p}{b_1-a_2}(e^{-p(a_1+a_2)t}-e^{-p(a_1+b_1)t})\\
\displaystyle \frac{\partial}{\partial t}P(\tau^B \leq t, \tau^B>\tau^A)= \frac{a_1(b_1+b_2)p}{a_1-b_2}(e^{-p(b_1+b_2)t}-e^{-p(a_1+b_1)t})
\end{array}
\right.
$$
Therefore, the marginal density of $\tau^A$ and $\tau^B$ are given by
$$
\left\{
\begin{array}{ll}
\displaystyle f_{\tau^A}(t,p)=\frac{b_1(a_1+a_2)p}{b_1-a_2}
\left(e^{-p(a_1+a_2)t}-e^{-p(a_1+b_1)t}\right)+a_1p e^{-p(a_1+b_1)t}\\
\displaystyle f_{\tau^B}(t,p)=\frac{a_1(b_1+b_2)p}{a_1-b_2}
\left(e^{-p(b_1+b_2)t}-e^{-p(a_1+b_1)t}\right)+b_1p e^{-p(a_1+b_1)t}.
\end{array}
\right.
$$
Comparing the results in Yu (2007),
the marginal density of $\tau^A$ and $\tau^B$ under
the standard reduced-form intensity-based model are, respectively, given by:
$$
g_{\tau^A}(t)=f_{\tau^A}(t,1) \quad {\rm and} \quad g_{\tau^B}(t)=f_{\tau^B}(t,1).
$$
Hence as the remark in Section 3, our proposed model is a generalization of the standard one.

\subsection{Basket Credit Default Swaps}

In this section we discuss the pricing of a basket CDS in the context of the multiple-default model described
in Section 3 and provide the sensitivity analysis of the CDS premium by varying key model parameters.
We consider a basket CDS contract that pays \$1 if $k$th-to-default out of a portfolio of reference entities
occurs prior to expiry date. To simplify our discussion, we assume that the payment (if any) occurs at expiration,
and that the buyer pays a premium at the initiation of the swap contract. With a constant risk-free interest rate $r$,
the premium on the $k$th-to-default CDS is given by:
$$
S_k=\exp\{-r T\}P(\tau^k \leq T) \ ,
$$
where $\tau^k$ is the $k$th-to-default time and T is the expiry date.

We assume the following default intensity for $n$ reference names:
$$
\lambda^i_t =X_t(1+b \sum_{j \neq i}1_{\{\tau_j \leq t\}}) \ , \quad i = 1, 2, \cdots, n \ ,
$$
and for all $i$, $p^i_t =1-e^{-cX_t}$.

Let $\tau^k$ denote the $k$th default time and $\tau^0$ is assigned to be 0.
Let
$$
\lambda^{(k)}_u = X_u (1+b(k-1))(n-(k-1))
$$
denote the $k$th-to-default rate as in Gu et al. (2011).
By Proposition 1, we have the following relations:
$$
 P\left(\tau^k-\tau^{k-1} >s \mid \sigma (\tau^{k-1}) \vee \mathcal{G_{\infty}}\right) = \exp\left\{-\int_{\tau^{k-1}}^{\tau^{k-1}+s} \lambda^{(k)}_u p_u du\right\}.
$$
Hence we have
$$
f_{\tau^k-\tau^{k-1} \mid \tau^{k-1}, \mathcal{G}_{\infty}}(s)
= \lambda^{(k)}_{\tau^{k-1}+s} p_{\tau^{k-1}+s}
\exp\left\{-\int_{\tau^{k-1}}^{\tau^{k-1}+s} \lambda^{(k)}_u p_u du\right\}.
$$
Then by convolution, given $\mathcal{G}_{\infty}$, we have the following recursive formula for
the conditional density function of $\tau^k$:
$$f_{\tau^k \mid \mathcal{G}_{\infty}}(t)
=\lambda^{(k)}_t p_t \int_0^t f_{\tau^{k-1} \mid \mathcal{G}_{\infty}}(u) \exp\left\{-\int_u^t\lambda^{(k)}_s p_s ds\right\} du.
$$
Let $Y_t =X_t p_t$ and by the recursive formula, assuming $b \neq 1/i$ for all $i=1,2,\ldots, n-1$, the PDF and CDF of $\tau^k$ respectively,
given the evolution of $(X_t)_{t \geq 0}$, are given by
$$
f_{\tau^k \mid \mathcal{G}_{\infty}}(t)
=\sum_{j=0}^{k-1}\alpha_{k,j}Y_t\exp\left\{-\beta_j\int_0^t Y_u du\right\}
$$
$$
P(\tau^k \leq t \mid \mathcal{G}_{\infty}) = \sum_{j=0}^{k-1}\frac{\alpha_{k,j}}{\beta_j}(1-\exp\{-\beta_j \int_0^t Y_u du\})
$$
where the coefficients are given by the following recursion:
$$
\left\{
\begin{array}{l}
\alpha_{k+1,j}= \left\{
\begin{array}{lll}
\displaystyle \frac{\alpha_{k,j} \beta_k}{\beta_k-\beta_j},& j=0,1,\ldots,k-1\\
\displaystyle -\sum_{u=0}^{k-1} \alpha_{k+1, u},& j=k
\end{array}
\right.\\
\displaystyle \beta_j = (n-j)(1+jb)
\end{array}
\right.
$$
where $\alpha_{1,0}=n$.
Hence,
$$
f_{\tau^k}(t) =\sum_{j=0}^{k-1}\alpha_{k,j} E (Y_t\exp\{-\beta_j\int_0^t Y_u du\})
$$
$$
P(\tau^k \leq t) = \sum_{j=0}^{k-1}\frac{\alpha_{k,j}}{\beta_j}(1-E(\exp\{-\beta_j \int_0^t Y_u du\})).
$$
If we assume $(X_t)_{t \geq 0}$ is a continuous time time-homogeneous Markov chain having $M$ states $\{x_1, x_2, \ldots, x_M\}$.
By using the results in Section 4, we find the swap premium $S_k$.
Indeed,  $Y_t=X_t p_t$ is also a continuous time time-homogeneous Markov chain having $M$ states $\{y_1, y_2, \ldots,y_M\}$, where $y_i=x_i(1-e^{-c X_i})$.
Then starting from state $x_i$, let $y=(y_1, y_2, \ldots, y_M)^T$, we have
$$
\begin{array}{ll}
P(\tau^k \leq t) &= \displaystyle \sum_{j=0}^{k-1}\frac{\alpha_{k,j}}{\beta_j}(1-E(\exp\{-\beta_j y^T T_i(t)\}))\\
& \displaystyle= \sum_{j=0}^{k-1}\frac{\alpha_{k,j}}{\beta_j}(1-\Psi_i(-\beta_j y, t)).
\end{array}
$$
Consequently we have
$$
S_k= e^{-rT}\sum_{j=0}^{k-1}\frac{\alpha_{k,j}}{\beta_j}(1-\Psi_i(-\beta_j y, T)).
$$
To give an example of the pricing of basket CDS, we set
the number of states of $(X_t)_{t \geq 0}$, $M=4$,
the states $x_i=0.1 i$, $p_{ij}=1/3$ for any $i \neq j$,
$v=(v_1,v_2,v_3,v_4)=(3,2,1,3)$.
We remark that the four states can represent the macroeconomic environment
as ``prosperous'', ``good'', ``neutral'', ``bad''.
We assume there are $10$ firms in the industry ($n=10$)
and we start from state $x_1$, ie, $X_0=x_1$,
the default free interest rate $r=0.05$, the expiry date $T=5$ years.
To examine the effect of contagion parameters $b$ and $c$,
we present the swap premium with the contagion parameters varying.
\begin{figure}
	\centering
		\includegraphics{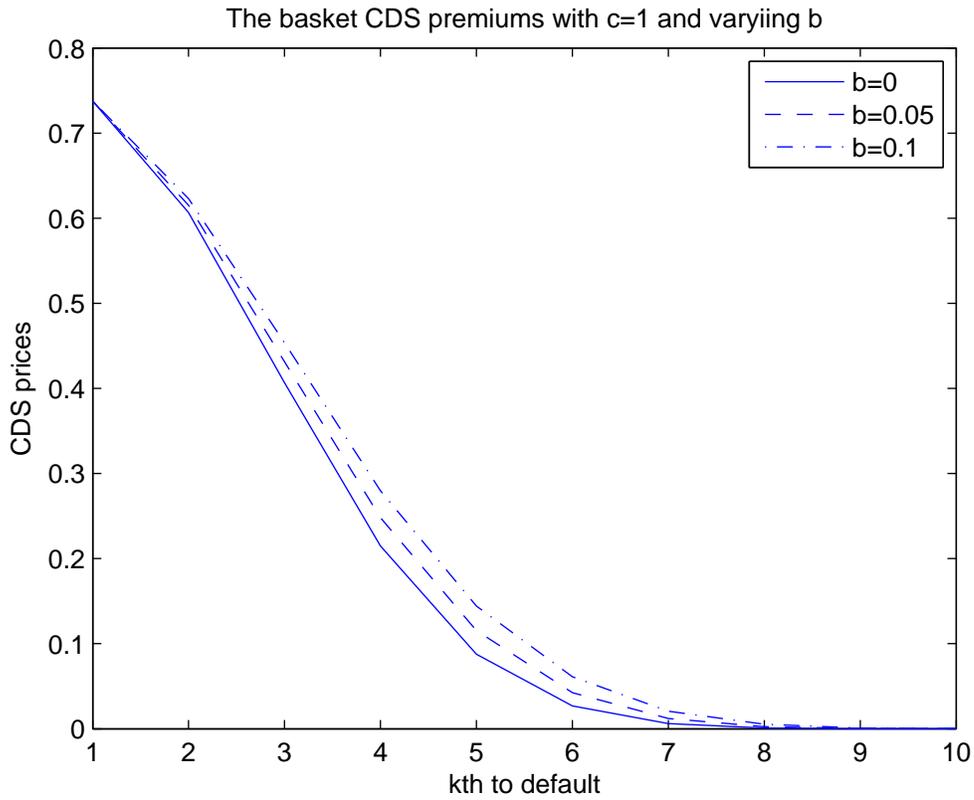} \\
		\includegraphics{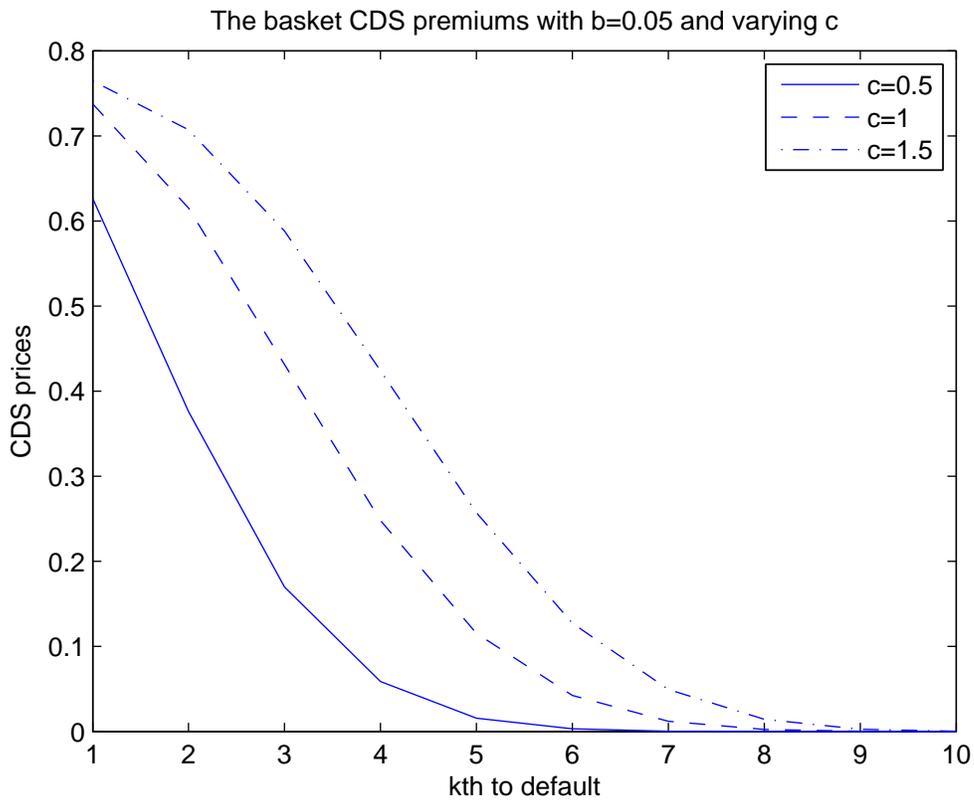}
\caption{The effect of the contagion parameters on the basket CDS premiums}
\end{figure}

As shown in Figure 1, an increase in the parameter $b$ indeed raises the premium for all $k$ as our intuition while the first-to-default swap price remains unchanged due to the contagion has no effect for the first-to-default time.
The premium increases as the parameter $c$ becomes bigger since the recovery from a ``trigger'' becomes more difficult in a weak macroeconomic environment.

\section{Concluding Remark}

In this paper, we introduce an extended reduced-form intensity-based model with ``trigger''
events, which captures an important feature in the market, namely the observable
events that trigger defaults.
Furthermore, we extend the model into a multi-name default model, with intensities driven by the history of the exogenous state process
representing the macroeconomic environment and the observed defaults.
A Markov chain model for the state process is also proposed to model the macroeconomic environment and the distribution of the occupation time of each states is deduced by solving a system of homogeneous ODEs.
We demonstrate the pricing of two defaultable contingent claims using the ordered default rate approach by Gu et al. (2011) with numerical examples.

There are still many outstanding issues for further research such as multi-state migration model that links loss sizes of a firm due to ``trigger'' events and its credit states which in turn affect payoffs to investors, and hybrid model of "trigger" events and default contagion with effective computation. 
We are currently investigating these problems.

\section{Appendix}

\subsection{Proof of Proposition 1}
\begin{proof}
Note that
$$
\begin{array}{lll}
P(\tau \leq s \mid \mathcal{G}_t)
&= & \displaystyle \sum_{i=1}^{\infty} P(\tau \leq s, K=i \mid \mathcal{G}_t)\\
&= & \displaystyle \sum_{i=1}^{\infty} E[ E[1_{\{\tau_i \leq s\}} 1_{\{K=i\}} \mid \sigma(\tau^1, \tau^2, \ldots, \tau^i) \vee \mathcal{G}_{\infty}]
\mid \mathcal{G}_t]\\
&= & \displaystyle \sum_{i=1}^{\infty} E[ 1_{\{\tau^i \leq s\}}p_{\tau^i} q_{\tau^{i-1}} q_{\tau^{i-2}} \ldots q_{\tau^1} \mid \mathcal{G}_t].
\end{array}
$$
By definition,
$$
P(\tau^{j+1}-\tau^{j}>u \mid \sigma(\tau^{j}) \vee \mathcal{G}_{\infty})= \exp\{-\int_{\tau^{j}}^{\tau^{j}+u} \lambda_v dv\}
$$
which implies that
$$
\frac{\partial }{\partial u}P(\tau^{j+1}-\tau^j >u \mid \sigma(\tau^j) \vee \mathcal{G}_{\infty})=  -
\lambda_{\tau^j+u}\exp\{-\int_{\tau^{j}}^{\tau^{j}+u} \lambda_v dv\} .
$$
Hence the joint density function of $\{\tau^1, \tau^2, \ldots, \tau^i\}$ given  ${\cal G}_{\infty}$ is:
$$
f_{\tau^1, \tau^2, \ldots, \tau^i \mid \mathcal{G}_{\infty}}(t_1, t_2, \ldots, t_i)=
\lambda_{t_1}\lambda_{t_2} \ldots \lambda_{t_i} \exp\{-\int_{0}^{{t_i}} \lambda_v dv\}
$$
Let $\{Z^j\}_{j=0}^{\infty}$ be a sequence of continuous stochastic processes
 $Z^j := \{ Z^j_u \}_{u \geq 0}$ such that
$$
Z^{j+1}_u :=\int_0^u Z^{j}_v \lambda_v q_v dv,
$$
where $Z^0_u=1$.
Thus
$$
\begin{array}{ll}
&\displaystyle  E[1_{\{\tau^i \leq s\}} p_{\tau^i} q_{\tau^{i-1}} q_{\tau^{i-2}} \ldots q_{\tau^1 } \mid \mathcal{G}_{\infty}]\\
= &\displaystyle  \int_0^s\int_0^{t_i}\int_0^{t_{i-1}}\ldots \int_0^{t_2} q_{t_1} q_{t_2} \ldots q_{t_{i-1}} p_{t_i} f_{\tau^1, \tau^2, \ldots, \tau^i \mid \mathcal{G}_{\infty}}(t_1, t_2, \ldots, t_i) d t_1 d t_2 \ldots d t_i\\
=& \displaystyle \int_0^s \lambda_u p_u \exp\{-\int_0^u \lambda_v dv\} Z^{i-1}_u du.
\end{array}
$$
Let $H_u = \sum_{i=1}^{\infty} Z^i_u$. Then by the monotone convergence theorem,
$$
H_u =\int_{0}^{u} (\sum_{i=1}^{\infty} Z^i_v +1)\lambda_v q_v dv = \int_0^u (H_v +1) \lambda_v q_v dv.
$$
Note that  a c\`adl\`ag process  has at most countably many discontinuities on a compact
interval.
Since $\{ \lambda_u \}_{u \geq 0}$ and $\{ p_u \}_{u \geq 0}$  are c\`adl\`ag processes, $H_u$ is continuous in $u \in \Re_+$ and differentiable except on a set with countable points
which has measure zero.
Consequently, except on an evanescent set,
$H_u$ satisfies the following Ordinary Differential Equation (ODE):
$$
\frac{d H_u}{du} = (H_u +1) \lambda_u q_u \ ,
$$
and $H_0 =0$.

Solving the ODE yields:
$$
H_u = \exp\left\{\int_0^u \lambda_v q_v dv\right\}-1.
$$
Therefore,
$$
\begin{array}{lll}
\displaystyle P(\tau \leq s \mid \mathcal{G}_t)
&=& \displaystyle  \sum_{i=1}^{\infty} E[E[1_{\{\tau^i \leq s\}} p_{\tau^i} q_{\tau^{i-1}} q_{\tau^{i-2}} \ldots q_{\tau^1 } \mid \mathcal{G}_{\infty}] \mid \mathcal{G}_t]\\
&=& \displaystyle \sum_{i=1}^{\infty} E[\int_0^s Z^{i-1}_u \lambda_u p_u \exp\{-\int_0^u \lambda_v dv\} du \mid \mathcal{G}_t]\\
&=& \displaystyle E[\int_0^s (1+H_u) \lambda_u p_u \exp\{-\int_0^u \lambda_v dv\} du \mid \mathcal{G}_t]\\
&=& \displaystyle E[\int_0^s \lambda_u p_u \exp\{-\int_0^u p_v \lambda_v dv\} du \mid \mathcal{G}_t]\\
&=&\displaystyle  1-\exp\{-\int_0^s p_u \lambda_u du\}.
\end{array}
$$
\end{proof} 

\subsection{Proof of Lemma 1}

\begin{proof}
Note that the conditional expectation is $0$ on the set $\{\tau \leq s\}$ and that the set $\{\tau > s\}$ is an atom in $\mathcal{I}_s$.
For any elements $\omega \in \mathcal{G}_t$ and $\tilde{\omega} \in \mathcal{H}_s$, due to the Markov property of the Poisson process,
$\tilde{\omega}$ and $\{\tau > t\}$ are independent conditional on $\omega$ and $\{\tau > s\}$. Consequently, using a version of the Bayes' rule and Proposition 1,
$$
\begin{array}{lll}
  E(1_{\{ \tau > t \}} \mid \mathcal{G}_t \vee \mathcal{H}_s \vee \mathcal{I}_s)
&=& 1_{\{ \tau > s \}}E(1_{\{ \tau > t \}} \mid \mathcal{G}_t \vee \mathcal{H}_s \vee \mathcal{I}_s)\\
&=& \displaystyle 1_{\{ \tau > s \}}\frac{P(\tau >t \mid \omega, \tilde{\omega})}{P(\tau >s \mid \omega, \tilde{\omega})}\\
&=& \displaystyle 1_{\{ \tau > s \}}\frac{P(\tau >t, \tilde{\omega} \mid \omega, \tau>s)}{P(\tilde{\omega} \mid \omega, \tau >s)}\\
&=& 1_{\{ \tau > s \}}P(\tau >t \mid \omega, \tau>s)\\
&=& \displaystyle 1_{\{ \tau > s \}}\frac{P(\tau >t \mid \omega)}{P(\tau >s \mid \omega)}\\
&=& \displaystyle 1_{\{ \tau > s \}}\frac{P(\tau >t \mid \mathcal{G}_t)}{P(\tau >s \mid \mathcal{G}_t)}\\
&=& \displaystyle 1_{\{ \tau > s \}}\frac{\exp\{-\int_0^t p_u \lambda_u du\}}{\exp\{-\int_0^s p_u \lambda_u du\}}\\
&=& 1_{\{ \tau > s \}}\exp\{-\int_s^t p_u \lambda_u du\}.
\end{array}
$$
\end{proof}

\subsection{Proof of Lemma 2}
\begin{proof}
For any $s< t$, using the results in Lemma 1,
$$
\begin{array}{lll}
  E[1_{\{t \geq \tau\}}-1_{\{s \geq \tau\}} \mid \mathcal{F}_s]
&=& E[1_{\{ \tau > s \}}-1_{\{ \tau > t \}} \mid \mathcal{F}_s]\\
&=& 1_{\{ \tau > s \}}(1-E[\exp\{-\int_s^t p_u \lambda_u du\}\mid \mathcal{F}_s]).
\end{array}
$$
Note that conditional on $\mathcal{G}_{\infty}$ and for $x >s$ the density of default time is given by
$$
\frac{\partial}{\partial x}P(\tau \leq x \mid \tau >s, \mathcal{G}_{\infty})=p_x\lambda_x \exp\{-\int_s^x p_u \lambda_u du\}.
$$
Hence,
$$
\begin{array}{lll}
  E[\int_s^t p_u \lambda_u 1_{\{u < \tau\}} du \mid \mathcal{F}_s]
&= &1_{\{ \tau > s \}} E[\int_s^t p_u \lambda_u 1_{\{u < \tau\}} du \mid \mathcal{F}_s] \\
&= &1_{\{ \tau > s \}} E[\int_s^t p_h \lambda_h \exp\{-\int_s^h p_u \lambda_u du\} (\int_s^h p_u \lambda_u du)  dh  \\
&& +\exp\{-\int_s^t p_u \lambda_u du\} (\int_s^t p_u \lambda_u du)  \mid \mathcal{F}_s] \\
&= &1_{\{ \tau > s \}}(1-E[\exp\{-\int_s^t p_u \lambda_u du\}\mid \mathcal{F}_s]).
\end{array}
$$
\end{proof}

\end{document}